\documentclass{article}
\usepackage{amsmath,amsfonts, amsthm,amssymb}
\newtheorem{definition}{Definition}
\newtheorem{lemma}{Lemma}
\newtheorem{theorem}{Theorem}
\newtheorem{proposition}{Proposition}
\newtheorem{corollary}{Corollary}
\newtheorem{remark}{Remark}
\title{Randomness of formal languages\\via automatic martingales}
\author{Birzhan Moldagaliyev}

\begin{document}
\maketitle
\begin{abstract}
We define a notion of randomness for individual and collections of formal languages based on automatic martingales acting on sequences of words from some underlying domain. An automatic martingale bets if the incoming word belongs to the target language or not. Then randomness of both single languages and collections of languages is defined as a failure of automatic martingale to gain an unbounded capital by betting on the target language according to an incoming sequence of words, or a text. The randomness of formal languages turned out to be heavily dependent on the text. For very general classes of texts, any nonregular language happens to be random when considered individually. As for collections of languages, very general classes of texts permits nonrandomness of automatic families of languages only. On the other hand, an arbitrary computable language is be shown to be nonrandom under certain dynamic texts. 
\end{abstract}

\section{Introduction}
Theory of algorithmic randomness~\cite{Downey2010} tries to capture properties which make mathematical objects appear random.
The theory mostly deals with infinite binary sequences. There are three major paradigms in theory of algorithmic
randomness: unpredictability, incompressibility and measure-theoretic typicalness. Viewed from automata-theoretic 
perspective, the theory mostly operates in higher levels of computability where full strength of Turing machine
is assumed. \\
As for randomness of formal languages, there are several ways one can approach the subject. The most direct 
way is to transform given language $L$ into the infinite binary sequence $\chi_L$ using some canonical 
transformation, with subsequent identification of $L$ with $\chi_L$. Indeed, in theory of algorithmic randomness, a set $A$ is said 
to be abc-random if and only if its characteristic sequence $\chi_A$ is abc-random, where abc is an arbitrary notion of randomness. Another possibility for measuring randomness of formal languages is to measure how well (or badly) given language $L$ can be approximated with members of simpler language classes from 
automata theory, including classes of regular and context-free languages. This line of research was pursued in 
the work of Yamakami~\cite{Yamakami2011}.\\
In this paper we propose an alternative definition of randomness for formal languages based on ideas from
automatic learning theory \cite{Jain2016} and theory of algorithmic martingales~\cite{Downey2010}. We consider a dynamic environment 
when words from underlying domain $D$ arrive in some order, which might be far from canonical, and a martingale
with some automata-theoretic properties is expected to bet whether incoming word is in $L$ or not. An 
automatic martingale is said to \emph{succeed} on a language $L$ if it reaches arbitrary high levels of capital in its run. A 
language is said to be \emph{random} if no automatic martingale succeeds on it.\\
Given framework can also be interpreted in the form of two-player game between automatic martingale and
adversary. The aim of the automatic martingale is to ensure unbounded growth of its capital. On the other hand, 
the adversary attempts to bound the growth of martingale's capital. To understand this interpretation better, let us 
consider a framework where order of incoming words is fixed to be length-lexicographic. Since the order is fixed,
the adversary has no role to play. On the other hand, in a framework where arbitrary orders or texts are allowed, the adversary has a full power to influence success of automatic martingale by providing somehow 'difficult' words to bet. By studying interaction between allowed classes of ordering and resulting randomness of languages, we hope to understand randomness of formal languages better.
\section{Background}
The paper assumes familiarity with basics of automata theory such as notions of regularity, syntactic classes and pumping lemma. Below we briefly present an additional background material which seems to be necessary for understanding the main sections of the paper. 
\subsection{Automatic Relations}
Automatic relations~\cite{Khoussainov2012} extend a notion of regularity from languages over simpler spaces to languages over product spaces. Suppose we are given a $k$-ary relation $R\subseteq (\Sigma^*)^k$, where an alphabet $\Sigma = \{0,1\}$ is assumed to be binary throughout the paper. One might ask if there is some automatic way of computing given relation. A notion of automatic relation attempts to do that. For that, we write given $k$-tuple $t=(t_1,t_2,\ldots,t_k)\in (\Sigma^*)^k$ in a block form:
\begin{equation}
\begin{bmatrix}
t_1\\
t_2\\
\cdots\\
t_k
\end{bmatrix}
\end{equation}
To make rows homogenous, shorter rows are filled with a special symbol, say $\#$. To process such blocks, one uses finite automata which read one symbol across all rows at a time. Given relation $R$ is said to be \emph{automatic} if there is a finite automaton $M$ recognizing it.
\paragraph{Automatic functions} Given a notion of automatic relation, it is straightforward to define a notion of \emph{automatic function}. A function $f:(\Sigma^*)^m \to (\Sigma^*)^n$ is called automatic if its graph forms an automatic relation, i.e. $graph(f) = \{(x,f(x))\mid x\in (\Sigma^*)^m \}$ is automatic.
\paragraph{Closure under first order definition} One of the most useful tools in working with automatic relations is their closure under the first order definition~\cite{Khoussainov2012}. This property can be states as follows:
\begin{proposition}[\cite{Khoussainov2012}]
Let $R$ be a first-order definable relation from given functions $(f_1,f_2,\ldots,f_n)$ and relations $(R_1,R_2,\ldots,R_m)$. If each of these functions and relations is automatic, then $R$ is also automatic. 
\end{proposition}
\paragraph{Examples} Let us give few examples of automatic relations so that a reader not familiar with them, might see them in action, so to speak. \\
Consider relation $R\subseteq (\Sigma^*)^2$ such that $R(x,y)=1 \Leftrightarrow \vert x \vert < \vert y \vert$, i.e. the first string should be strictly shorter than the second one. To build a finite automaton $M=(S,\Sigma,f,s_0,F)$ recognizing $R$, let us have: $S = \{s_0,s_1\}$, $F=\{s_1\}$ and a transition function $f$ as follows:
\begin{align}
f(s_0,(a,b)) = 
\begin{cases}
s_1, & \text{ if } a=\#,\, b\neq \#\\
s_0, & \text{ otherwise }
\end{cases}
&&
f(s_1,(a,b)) = s_1, \text{ for all } (a,b)
\end{align}
In a similar vein, one could show automaticity of \emph{lexicographic} order, $<_{lex}$, on $\Sigma^*$ generated by the canonical ordering $(0<1)$ of the underlying alphabet. The lexicographic order induces so called \emph{length-lexicographic order} $<_{ll}$ on $\Sigma^*$ defined as:
\begin{equation}
x<_{ll}y \Leftrightarrow \vert x \vert < \vert y \vert \text{ or } (\vert x \vert = \vert y \vert \text{ and } x<_{lex}y)
\end{equation}	
Automaticity of this order follows from first-order definability property described earlier. The length-lexicographic order provides many benefits, with its linearity being one of the most notable ones. As a final example, let us state a following well-known fact:
\begin{proposition}
Any regular language $D$ can be embedded in $(\Sigma^*)^k$ for some $k$. Moreover, the embedding is automatic.
\end{proposition}
\begin{proof}
Let $\Gamma$ be an underlying alphabet of $D$ of size $n$. Then there is the smallest integer $k$ such that $n\le 2^k$. Let $\phi$ be an injective map $\phi:\Gamma \to \Sigma^k$. We can naturally extend $\phi$ to be the mapping between $D$ into $(\Sigma^*)^k$ by applying $\phi$ letter-wise. In other words, the image of the word $x=x_0x_1\ldots x_{r-1}$ is given by $\phi(x)=\phi(x_0)\phi(x_1)\ldots \phi(x_{r-1})$. Clearly, this mapping is injective and automatic, hence it is an automatic embedding. Furthermore, the image of $D$ can be shown to be a regular language, simply because:
\begin{equation}
y\in Im(D) \Leftrightarrow \exists x\in D(\phi(x)=y)
\end{equation}
is a first-order definable formula. This completes the proof.
\end{proof}
Above fact tells that any regular language $D$, no matter how complicated, might be considered as a regular language inside $(\Sigma^*)^k$ for some $k$. This in turn shows universality of regular domains of the type $(\Sigma^*)^k$, which are used as bases for definition of automatic martingales, coming shortly. 
\subsection{Automatic Structures}
A notion of the structure plays an important role in mathematics. A mathematical structure is a set with operations and relations defined on it. Automata theory can be used to study some of mathematical structures. A structure $\mathcal{M} = (M,f_1,f_2,\ldots,f_n,R_1,R_2,\ldots,R_m)$, where $M$ is an underlying set, $f_i$'s are functions and $R_j$'s are relations is called automatic if:
\begin{itemize}
\item $M$ is a regular language;
\item Each $f_i$ is an automatic function;
\item Each $R_j$ is an automatic relation.
\end{itemize}
Moreover, structures which are isomorphic to automatic structures are also called automatic. 
\subsection{Dyadic rationals, $\mathbb{Q}_2$}
A dyadic rational is a rational number given in the form $\frac{a}{2^b}$ where $a$ is an integer and $b$ is a natural number. The collection of all dyadic rational numbers $\mathbb{Q}_2$ form a commutative ring with standard addition and multiplication. Moreover, $\mathbb{Q}_2$ inherits metric topology from $\mathbb{R}$. This allows to define limiting processes which might end up outside of $\mathbb{Q}_2$, though. 
\paragraph{Automatic presentation} Any element $a\in \mathbb{Q}_2$ can be presented as:
\begin{equation}
a = (-1)^s(\sum_{i\in\mathbb{Z}}a_i2^i)
\end{equation}
with coefficients $a_i\in\{0,1\}$ and $s\in\{0,1\}$, where finitely many of them being nonzero. Let $n,m$ be the greatest and smallest nonzero indices respectively, then $a$ can be presented in the following form:
\begin{equation}
a = 
\begin{bmatrix}
a_0 & a_1 & a_2 & \ldots & a_n & \ldots & \# \\
s & a_{-1} & a_{-2} & \ldots  & \ldots & \ldots & a_m
\end{bmatrix}
\end{equation}
Thus, the above presentation of dyadic rationals is identified with a product space $(\Sigma^*)^2$, which is clearly regular language. Moreover, observe that addition of two dyadic rationals is automatic in this representation. A simple reason for this is a familiar addition with carrying. Each time the presence of a carry should be stored in the memory, which is something finite automata are capable of doing. For fuller exposition of this fact one can refer to \cite{Nies2007} 
\paragraph{Multiplication in this presentation} Let us observe how multiplications by $2$ and $2^{-1}$ are performed in this presentations. They correspond to shifting each row to the left or right. More explicitly:
\begin{align*}
2a =&
\begin{bmatrix}
a_{-1} & a_0 & a_1 & \ldots & a_n  & \ldots & \# \\
s & a_{-2} & a_{-3} & \ldots  & \ldots & \ldots & a_m
\end{bmatrix}\\
2^{-1}a = &
\begin{bmatrix}
a_1 & a_2 & a_3 & \ldots & a_n & \ldots & \# \\
s & a_0 & a_{-1} & \ldots & \ldots & \ldots & a_m
\end{bmatrix}
\end{align*}

Since shifts by left and right can be executed by automatic functions, multiplications by $2$ and $2^{-1}$ are automatic. Due to closure under first-order definition, we have that any composition of above multiplications paired with additions is also automatic. Hence, a multiplication by any fixed dyadic rational is automatic. Observe that a multiplciation of two dyadic rationals is not automatic in this presentation. Let us state this well-known fact together with a short proof:
\begin{proposition}
Multiplication is not automatic in the given presentation of dyadic rationals
\end{proposition}
\begin{proof}
Suppose contrary, and let $p$ be a pumping constant corresponding to the given automatic multiplication. Consider $a = 2^p$, with a representation of length $p+1$. Observe that $a\times a = 2^{2p}$ has the presentation of length $2p+1$. This means that second half of $a^2$ with length $p$ can be pumped up, thus violating uniqueness of multiplication. This leads to a contradiction.
\end{proof}
\paragraph{Order relation} Let us now turn our attention to ordering on $\mathbb{Q}_2$. Our aim is to show that the standard order $<$ on $\mathbb{Q}_2$ is automatic. To show that, we need some intermediate facts. Let us consider following relations defined on $\mathbb{Q}_2$:
\begin{enumerate}
\item $z(a)$: checks if $a=0$;
\item $p(a)$: checks if $a>0$;
\item $l(a,b)$: given a pair of elements $a,b$, this relation checks if $a<b$.
\end{enumerate}
\begin{proposition}
In given representation of $\mathbb{Q}_2$, all of $z(a),p(a)$ and $l(a,b)$ are automatic relations
\end{proposition}
\begin{proof}
1. Let us assume abovementioned presentation for $a\in\mathbb{Q}_2$. Observe that
\begin{equation}
z(a)=1 \Leftrightarrow (a_1a_2\ldots a_n\in 0^*)\text{ and }(a_{-1}a_{-2}\ldots a_{-m}\in 0^*)
\end{equation}
Since both relations on the right are automatic and automaticity is closed under intersections, we infer that $z(a)$ is indeed automatic.\\
2. Verifying automaticity of $p(a)$ relies on the previous result. Let us observe that:
\begin{equation}
p(a)=1 \Leftrightarrow s=0 \text{ and } z(a)=0
\end{equation}
Since both relations on the right hand side are automatic, we infer that $p(a)$ is automatic.\\
3. To show automaticity of $l(a,b)$ we make use of first-order closure of automatic relations. Namely,
\begin{equation}
l(a,b)=1 \Leftrightarrow \exists c[(p(c)=1)\text{ and }(a+c=b)]
\end{equation}
So, it follows that $l(a,b)$ is an automatic relation. This completes the proof of given proposition.
\end{proof}
The following well-known proposition sums up our observations regarding automatic presentations of dyadic rationals.
\begin{proposition}
A structure $(\mathbb{Q}_2,+,<,c(a_1),\ldots,c(a_k))$ is automatic, where $c(a_i)$ denotes a multiplication by $a_i$ with $a_i\in \mathbb{Q}_2$.
\end{proposition}
\subsection{Automatic Learning Theory}
The notion of automatic learning~\cite{Jain2012} was developed as an automata-theoretic equivalent of algorithmic learning theory developed by Gold~\cite{Gold1967}. In general, algorithmic learning theory deals with following basic parts. Given some collection of languages $\mathcal{L}$, a \emph{text} is an infinite sequence of words belonging to some $L\in\mathcal{L}$. The objective is to somehow infer a corresponding index of $L$ in $\mathcal{L}$ based on the given text. Classical algorithmic learning theory deals with a collection of recursively enumerable languages as a base and computable functions as tools of inference. On the other hand, automatic learning theory deals mainly with automatic families of languages~\cite{Jain2016} as a base and automatic functions as state updates, where the definition of automatic family reads as follows:
\begin{definition}[Automatic family \cite{Jain2016}]
A collection of languages $\mathcal{L} = \{L_e \}_{e\in E}$ is called automatic family if:
\begin{itemize}
\item $E$ is a regular language;
\item $\{(x,e)\mid x\in L_e \}$ forms an automatic relation.
\end{itemize}
\end{definition}
A state of a learner is given by a pair of memory and hypothesis, $S=M\times E$. An automatic learner is an automatic function $f$ which updates a current state given an upcoming word, $f:S\times D\to S$, where $D$ is a some automatic domain words are drawn from. It is desirable that an induced sequence of hypotheses should converge to a correct index. We are going to borrow two ideas from this theory. First, we are going to understand languages by a sequence of words from underlying domain, or simply by a text. Secondly, we are going to employ the idea of automatic transition functions.  
\subsection{Algorithmic martingales}
An algorithmic martingale from theory of algorithmic randomness \cite{Downey2010} is a function $d:\{0,1 \}^*\to \mathbb{R}$ satisfying so-called fairness condition:
\begin{equation}
2d(x) = d(x0)+d(x1),\quad \forall x\in\{0,1\}^*
\end{equation}
In theory of algorithmic randomness $d$ is interpreted as a betting strategy which allocates its available capital, $d(x)$, between two scenarios: next bit is $0$ and next bit is $1$. In probabilistic interpretation, the fairness condition ensures that expected capital value at the next stage is equal to current capital value under equiprobable distribution on $\{0,1\}$. Both interpretations are useful to have in mind. In constructing martingales, it is useful to think of a martingale as a dynamic process. On the other hand, it is useful to think of a martingale as fixed object in proving nonexistence of martingales satisfying particular relations. 
\paragraph{Definition of random objects} A martingale $d$ is said to \emph{succeed} on an infinite sequence $X$ if
\begin{equation}
\limsup_n d(X[n]) = \infty
\end{equation}
or $\{d(X[n]) \}_n$ is unbounded, where $X[n]$ refers to prefix of $X$ of length $n$. An infinite sequence $X$ is said to be \emph{abc-random} if there is no martingale satisfying condition \emph{abc} which succeeds on $X$, where \emph{abc} is arbitrary randomness notion. We are going to borrow two ideas from this theory. Firstly, we impose fairness conditions. Secondly, we define random objects in a spirit similar to the above definitions. 
\section{Definitions}
\subsection{State space}
In order to define a notion of an automatic martingale, we need to start with a notion of a state which consists of capital and memory. Capital values are given by the set $C$, in our case $C=\mathbb{Q}_2$. As for memory, we assume its elements to reside in a product space $(\Sigma^*)^i$ for some $i\ge 0$, which is assumed to be fixed after initial setting. A state space is given by the product $S=C\times M$. Note that presentation of the state space corresponds to convolution of capital and memory presentations. Given newly formed state space, we introduce projection map, $\pi:S\to C$, which projects state value to its capital value, i.e. $\pi(s)=c$ given $s=(c,m)$. Observe that $\pi$ is an automatic function in given presentation. Moreover let $\mathcal{S}=S^{\mathbb{N}}=\{\phi:\mathbb{N}\to S\}$ and $\mathcal{C}=C^{\mathbb{N}}=\{\psi:\mathbb{N}\to C\}$ be collections of infinite sequences of states and capital values respectively. Then $\pi$ can naturally be extended to the map between $\mathcal{S}$ and $\mathcal{C}$ as follows:
\begin{equation}
\pi(\phi)(n) = \pi(\phi(n))
\end{equation}
This map is going to be useful for us later on. 
\subsection{Automatic martingales}
An automatic martingale $f$ is meant to update a current state given an incoming data point satisfying some fairness conditions. Let us first clarify what do we mean by a data point. A data point $t$ is either word from underlying domain with the membership label in a target language or a special skip symbol $\#$, i.e. $t\in (D\times \Sigma)\cup\{\#\}=T$. We assume that the domain $D$ is given by some regular language. Automatic martingales are intended to describe some kind of randomness of languages inside the domain $D$. An \emph{automatic martingale} $f:S\times T\to S$ is an automatic function satisfying a given fairness condition:
\begin{align}
2\pi(s) &= \pi(f(s,x,0)) + \pi(f(s,x,1))& \text{ if } x\in D\\
\pi(s) &=\pi(f(s,t)) &\text{ if } t=\#
\end{align}
Above conditions can be interpreted as follows. Given a current state $s$, automatic martingale distributes its capital fairly between outcomes $x\in L(L(x)=1)$ and $x\not\in L(L(x)=0)$ for any word $x\in D$, where $L$ is a language under investigation. Second condition says that special symbol $\#$ does not alter capital value of the automatic martingale. The presence of special symbol $\#$ is attributed to the legacy from algorithmic learning theory, where sometimes incoming data might be void. An automatic martingale is assumed to act in infinite sequence of data points. To generate an infinite sequence of data points, we label some text $X:\mathbb{N}\to D\cup\{\#\}$ with a membership in $L$ as follows:
\begin{equation}
Z(n) = 
\begin{cases}
(X(n),L(X(n))), & \text{ if }X(n)\neq \#\\
\#, & \text{ otherwise}
\end{cases}
\end{equation}
Such labeling of the text $X$ with respect to the language $L$ is written as $Z=X\circ L$. In order to be valid, a sequence of data points should have infinitely many labeled words. Formally, it can be written as:
\begin{equation}
\forall n\, \exists m>n\, Z(m)\neq \#
\end{equation}
where $Z:\mathbb{N}\to T$ is a sequence of data points, further referred as a \emph{stream}. 
\subsection{Action of martingales and randomness of languages}
Let us now describe an action of automatic martingales on streams. Given an automatic martingale $f$ and some starting state $s_0$, we say that they form a \emph{setup} $d=(f,s_0)$. As data points arrive, automatic martingale $f$ updates its states. Assuming that a target language is $L$ and a text under consideration is $X$, the resultant stream is $Z=X\circ L$. The stream induces a sequence of states $\phi^S\in\mathcal{S}$ as follows
\begin{align}
\phi^S(0) &= s_0\\
\phi^S(n+1) &= f(\phi^S(n),Z(n))
\end{align}
In this way, $d$ can be regarded as a map from collection of streams $\mathcal{Z}$ to $\mathcal{S}$. Composing the map $d$ with the projection operator $\pi$ we obtain the map $d^{\pi}:\mathcal{Z}\to \mathcal{C}$. If $\pi^d(Z)$ turns out to be unbounded sequence of dyadic rationals, i.e. $\limsup(d^{\pi}(Z))=\infty$, we say that $d$ \emph{succeeds} on the language $L$ under the text $X$. Similarly, we say that a setup $d$ succeeds on a collection of languages $\mathcal{L}$ under a collection of texts $\mathcal{X}$ if $d$ succeeds on every $L\in\mathcal{L}$ under any text $X\in\mathcal{X}$.  A \emph{normed setup} refers to a setup $d=(f,s_0)$ such that $\pi(s_0)=1$. Finally, we are ready to define randomness for languages.
\begin{definition}[Randomness of Formal Languages]
A collection of languages $\mathcal{L}$ is said to be random under a class of texts $\mathcal{X}$ if there is no normed setup succeeding on $\mathcal{L}$ under $\mathcal{X}$.
\end{definition}
\begin{remark}
If a collection $\mathcal{L}=\{L\}$ consists of a single language, we simply refer to it as a language $L$. Similarly, a singleton class $\mathcal{X}=\{X\}$ is referred to as a text. 
\end{remark}

\paragraph{Possible texts} In the above definition of randomness there is a clear dependence on the underlying text which generates the stream. The randomness of given language might vary depending on the text under consideration. Below we present examples of classes of texts:
\begin{itemize}
\item The class of all valid texts, $\mathcal{T} = \{X:\mathbb{N}\to D\cup\{\#\}\mid \forall n\,\exists m>n\, (X(m)\neq\#)\}$. 
\item \emph{Infinite Range}, \(\mathcal{I}\). This class consists of texts $X$ such that range of $X$ is infinite. Put formally, \(\mathcal{I}=\{X\in \mathcal{T}\mid\ \text{Range of }X \text{ is infinite}\}\).
\item \emph{Exhaustive}, \(\mathcal{E}\). This class consists of texts which exhaust elements of the underlying domain $D$. Put formally, \(\mathcal{E}=\{X\in \mathcal{T}\mid \forall y\in D, \exists\, n\in \mathbb{N}\,(X(n)=y) \}\).
\item \emph{Repetition-free}, \(\mathcal{R}\). This class consists of texts with no repetitions of domain elements. Put formally, \(\mathcal{R}=\{X\in \mathcal{T}\mid X(n)=X(m)\neq \# \Rightarrow n=m \}\).
\item \emph{Ordered}, $\{X_{ll}\}$. This collection consists of a single text, which outputs all strings of $D$ in increasing  length-lexicographic order. 
\item \emph{Dynamic}, \(\mathcal{D}\). This collection consists of texts generated dynamically by automatic martingale itself. Let $g:S\to D\cup\{\#\}$ be an automatic function mapping from state space $S$ to underlying domain $D$ with the additional symbol. We say that a text $X:\mathbb{N}\to D\cup\{\#\} $ is generated by automatic function $g$ if $X(n) = g(s_n)$ where $s_n$ refers to a state of automatic martingale at stage $n$. We say that a text $X$ is \emph{dynamic} if it is generated by some automatic function. 
\end{itemize}

\section{Properties}
In this section we are going to study properties of the notions defined above. First we are going to study randomness of individual languages. Afterwards, we are going to investigate randomness of collections of languages. 
\subsection{Randomness of individual languages}
We are going to observe that randomness of individual languages depends very much on the class of text under consideration. We start our investigations from the most general class of texts and end with dynamic texts, which are set by automatic martingales themselves. 
\paragraph{Class of all texts} Let us consider randomness under the most general class $\mathcal{T}$ of all possible texts. 
\begin{theorem}
A language $L\subseteq D$ is random under the class $\mathcal{T}$ if and only if it is not regular.
\end{theorem}
\begin{proof}
\textit{Forward direction}\\
We show validity of the contrapositive statement. Given a regular language $L$, we want to show that $L$ is not random under $\mathcal{T}$. In other words, we need to exhibit an automatic setup $d=(f,s_0)$ succeeding on $L$ under any text $X\in \mathcal{T}$. The regularity of $L$ allows to construct such setup with considerable ease. In this case, we do not need to use any memory, so $M=\Sigma^*$ should suffice. Let $s_0=(1,\varepsilon)$ is a starting state with unit capital value and empty string as memory value. We set an automatic martingale $f$ as follows:
\begin{align}
f((c,m),x,b)=
\begin{cases}
(\frac{3}{2}c,m) & \text{ if } L(x)=b\\
(\frac{1}{2}c,m) & \text{ otherwise}
\end{cases}
&&
f((c,m),\#)=(c,m)
\end{align}
Clearly, above function is automatic, due to the fact that the case distinction $L(x)=b$ is automatic thanks to regularity of $L$. It is clear that given setup $d=(f,s_0)$ succeeds on $L$ with respect to any text $X$, because $X$ contains infinitely many labeled elements from $D$.\\
\textit{Converse direction}\\	
Going for contrapositive statement, we show that nonrandomness under $\mathcal{T}$ implies regularity. Suppose $L$ is some nonrandom language under $\mathcal{T}$, so  there is an automatic setup $d=(f,s_0)$ succeeding on $L$ under any $X\in \mathcal{T}$. The idea is to construct an adversarial text $X$ which exhibits a regularity of $L$. We construct $X$ dynamically depending on the behaviour of $d$ under the current prefix of $X$. Suppose we have constructed $X$ up to prefix of length $n$, i.e. $X[n]=x_0,\, x_1,\ldots, x_{n-1}$. Let $s_n$ be a state obtained from processing $X\circ L[n]$ from initial state $s_0$. Consider a collection of words which do not increase the capital at the next stage:
\begin{equation}
D_n = \{x\in D\mid \pi(f(s_n,x,L(x)))\le \pi(s_n) \}
\end{equation}
If $D_n$ is nonempty we choose any element from it to append to $X[n]$. Observe that if there are infinitely many $n$'s such that $D_n\neq \emptyset$, then we succed to construct a text $X$ such that $\pi(s_n)\le 1$ for all states $s_n$ visited. This contradicts the fact that $L$ is nonrandom. Hence $D_n = \emptyset$ for some $n$. This means that
\begin{equation}
\pi(f(s_n,x,L(x)))>\pi(s_n) \text{ for all }x\in D
\end{equation}
Thus, we can make use of the fairness condition to check a membership in $L$:
\begin{equation}
x\in L \Leftrightarrow f(s_n,x,1)>f(s_n,x,0)
\end{equation}
As this condition involves only automatic function $f$ and fixed state $s_n$, it is automatic, hence $L$ is a regular language. 
\end{proof}
\paragraph{Class of exhaustive texts} From general class $\mathcal{T}$, let us now reduce the class under consideration to the class of exhaustive texts, $\mathcal{E}$. Before that, let us recall the definition of immunity as in theory of complexity and computability. Given a class of languages $\mathcal{L}$, a language $M$ is said to be $\mathcal{L}$-immune if:
\begin{itemize}
\item $M$ is infinite;
\item There is no infinite language $L\in\mathcal{L}$ contained in $M$.
\end{itemize}
The idea is that $M$ somehow avoids containing infinite members of $\mathcal{L}$. A language $M$ is called $\mathcal{L}$-bi-immune if both $M$ and $D\setminus M$ are $\mathcal{L}$-immune, where $D$ is the underlying ambient domain in which $M$ resides. We denote the class of all regular languages as REG. 
\begin{theorem}
Suppose that a language $L\subseteq D$ is random under the class of exhaustive texts, $\mathcal{E}$. Then $L$ is REG-bi-immune. 
\end{theorem}
\begin{proof}
The idea is similar to the proof of Theorem 1. Going for contrapositive statement, suppose $L$ is not REG-bi-immune. Without loss of generality assume that there is an infinite regular language $R$ such that $R\subseteq L$. We want to show that $L$ is not random under $\mathcal{E}$. In other words, there should be an automatic setup $d=(f,s_0)$ succeeding on $L$ under any $X\in\mathcal{E}$. A regularity of $R$ allows to construct $d$ directly. Again, we do not need memory in this case, so setting $M=\Sigma^*$ suffices. Let $s_0=(1,\varepsilon)$ as in previous proof. We define automatic martingale $f$ as follows:
\begin{align}
f((c,m),x,b)=
\begin{cases}
(c,m) & \text{ if } R(x)=0\\
(\frac{3}{2}c,m) & \text{ if } R(x)=1,\, b=1\\
(\frac{1}{2}c,m) & \text{ otherwise}
\end{cases}
&&
f((c,m),\#)=(c,m)
\end{align}
It is clear that given automatic setup $d$ succeeds on $L$ under any exhaustive text. The case where $R\subseteq L^c$ can be dealt with in a similar manner. 
\end{proof}

Now we would like to present an application of this theorem. First, we need to verify some well-known facts presented as lemmas.
\begin{lemma}
Given a context-free language $L$ and fixed words $u$ and $v$, consider $L'=\{w:uwv\in L\}$. Then $L'$ is still context-free.
\end{lemma}

\begin{proof}
To start with, consider a case where $u = a\in \Sigma$ and $v=\varepsilon$. Let $G=(V,\Sigma,R,S)$ be a context-free grammar in Chomsky normal form generating $L$. We are going to construct a context-free grammar $G'=(V',\Sigma,R',\dot{S})$ generating $L'$, thus showing that $L'$ is context-free. As for nonterminals, $V' = \{X,\dot{X}\mid X\in V,X\neq S\}\cup\{\dot{S}\}$. As for a new set of production of rules, we set:
\begin{itemize}
\item For each rule of type $X\to YZ$, we keep it and add a rule $\dot{X}\to \dot{Y}Z$;
\item For a rule $X\to a$, we keep it while adding $\dot{X}\to \varepsilon$.
\end{itemize}
Clearly, the new grammar $G'$ generates $L'$. A case of $u=\varepsilon$ and $v=a$ can be handled in a similar manner. As initial transformation from $L$ into $L'$ can be realized as a composition of transformations as given above, we conclude that $L'$ is indeed context-free. 
\end{proof}

\begin{lemma}
Any context-free language $L\subseteq D$ is not REG-bi-immune relative to a regular domain $D$.
\end{lemma}

\begin{proof}
Let $p$ be a pumping constant of the regular domain $D$. Given some word $x\in D$ of length at least $p$, it is possible to write it as $x=uvw$ such that $uv^*w\subseteq D$. There are two cases to consider. If $M=L\cap uv^*w$ happens to be a finite set, then we have an infinite regular language outside $L$ in the form of $uv^*w\setminus L$. So, assume that $M$ is an infinite set. Since intersection of a context-free language with a regular language is context-free, $M$ is a context-free language. It can be viewed as $M=u\cdot N \cdot w$, where $N\subseteq v^*$. By Lemma 1, $N$ is a context-free language. Let $q$ be a pumping constant of pumping lemma for context-free languages corresponding to $N$. Let $y\in N$ with length at least $q$. According to the pumping lemma, it is possible to write $y$ as $y=abcde$, where $|bd|\ge 1$, such that $ab^ncd^ne\in N$ for $n\ge 0$. It is clear that $|bd| = k|v|$ for some $k$. Hence, $ab^ncd^ne=v^m\cdot (v^k)^n$ for some $m\ge 0$. Collecting those words we form the infinite regular language $S=\{ab^ncd^ne:n\ge 0\}\subseteq N$. Appending the prefix $u$ and the siffix $w$, we have that $u\cdot S \cdot w$ is an infinite regular language inside $L$. This shows that $L$ is not REG-bi-immune. 
\end{proof}

\begin{corollary}
Arbitrary context-free language $L\subseteq D$ is not random with respect to the text class $\mathcal{E}$.
\end{corollary}

\paragraph{Dynamic texts} In the introduction, we have mentioned that given framework can be considered as a two-player game between automatic martingale $f$ and adversary. The automatic martingale aims to increase its capital, while the adversary wants it to stay bounded. Variety of possible texts gives power the adversary to play against the automatic martingale. In this sense, the larger the allowed class of texts, the more power adversary has. But what happens, when power to choose the text shifts from the adversary to the automatic martingale? In other words, we are going to investigate the class of dynamic texts. 
\begin{theorem}
Let $L$ be a computable language. Then it is nonrandom under some dynamic text.
\end{theorem}
\begin{proof}
Let $L$ be a computable language and $\phi_e$ be a Turing machine computing membership of words in $L$. We need to show an existence of a normed setup $d=(f,s_0)$ such that it succeeds on $L$ under some dynamic text. The idea is to simulate computations of $\phi_e$ on some specified inputs. As long as the simulation goes on, the text generating function $g$ sets next element of the text to be $\#$. When simulation completes, the automatic martingale uses its capital to bet on the corresponding outcome. \\
Having described the desired automatic martingale, let us proceed with its construction. We construct automatic setup $d$ and automatic function $g$ generating desired dynamic text at the same time. We set memory space $M = (\Sigma^*)^3$, which comprises of three values: input tape value $m^{I}$, work tape value $m^{W}$ and output tape value $m^{O}$. A reader might observe that given memory corresponds to the standard structure of Turing machine. As for the starting state $s_0$, we set $s_0 = (1,d_0,\varepsilon,\varepsilon)$, where $d_0=\min_{ll}(D)$ is a length-lexicographic minimum of the domain $D$. As for the automatic martingale $f$ and the automatic function $g$, they act as follows. We need to consider two cases:\\
1. The output is not computed, $m^{O} = \varepsilon$. In this case, $g$ sets the next value of the text to be $\#$, because the computations is not finished yet. As for $f$, it updates the value of $m^{W}$ according to the transition of working tape corresponding to $\phi_e$, leaving other components of memory and capital unchanged.\\
2. The output is computed, $m^{O}\in\{0,1 \}$. In this case, $g$ sets the next word to be $m^{I}$, because computation of $L(m^{I})$ has been finished. As for automatic martingale $f$ places all current capital to the outcome corresponding to $m^{O}$. At the same time, it clears the working tape, and assigns new input value. More formally:
\begin{equation}
f((c,m^{I},m^{W},m^{O}),(m^I,b))=
\begin{cases}
(2c,succ_{ll}(m^{I}),\varepsilon,\varepsilon), &\text{ if } b = m^{O}\\
(0, succ_{ll}(m^{I}),\varepsilon,\varepsilon), &\text{ otherwise}
\end{cases}
\end{equation}
Let us perform a quick verification exercise for automaticity of both $f$ and $g$. Firstly, given a state $s=(c,m)$ comprising capital and memory, case distiction $m^{O}=\varepsilon$ is automatic. Each individual transition between configurations of $\phi_e$ is automatic. Finally, since length-lexicographic order is automatic, computing length-lexicographic successor as in $succ_{ll}(m^{I})$ is automatic. Since both $f$ and $g$ are given in terms of these and simpler functions, both of them are automatic. 
\end{proof}

\paragraph{Random languages} So far we have discussed instances of nonrandom languages under specific texts. Now it is time to discuss instances of random languages under certain text types. To construct a random language we need to ensure failure of every automatic martingale on that language. It is rather difficult task, for there are various kinds of automatic martingales, which capable of targeting different features of the given language. Fortunately, history of computability theory and complexity theory is rich with constructions of languages meeting infinite number of requirements. One of the most well-known techniques used for this purpose is \emph{diagonalization} technique. Let us demonstrate diagonalization technique in our settings. For this let us first define some arithmetic operations on the space of setups $\mathcal{D}=\mathcal{F}\times S$, where $\mathcal{F}$ refers to the space of all automatic martingales. Given setups $d_1=(f_1,s_1)$, $d_2=(f_2,s_2)$ and a scalar $c\in \mathbb{Q}_2$ we define operations of addition and scalar multiplication as follows. Given any stream $Z$, we wish to have: 
\begin{align}
(d_1+d_2)^{\pi}(Z)&= d_1^{\pi}(Z) + d_2^{\pi}(Z)\\
(cd_1)^{\pi}(Z)&=cd_1^{\pi}(Z) 
\end{align}
where additon and multiplication on the right hand side refer to point-wise addition and scalar multiplication on $\mathcal{C}$, space of sequences with values in $\mathbb{Q}_2$. A natural question to ask if these operations are closed in the space of setups. This issue is addressed in the next theorem.
\begin{theorem}
Both operations of addition and scalar multiplication given above are closed on the space of setups.
\end{theorem}
\begin{proof}
\textit{Addition}\\
Let $S_1$ and $S_2$ be state spaces of $d_1$ and $d_2$ respectively. Assume that $S_1$ is $n$-dimensional, while $S_2$ is $m$-dimensional, i.e. $S_1=(\Sigma^*)^n$ and $S_2=(\Sigma^*)^m$. Given two states $p\in S_1$ and $q\in S_2$, we define a corresponding state of desired automatic martingale as $(\pi(p)+\pi(q), (p,q))$. Note that a pair $(p,q)$ is realized as a convolution when it comes to automatic presentations. This means that memory $M$ for desired automatic martingale is given as a product $S_1\times S_2$, which makes it $m+n$-dimensional. Having defined state space, we are left to define an automatic martingale and a starting state. As for the automatic martingale $f$, it is defined as
\begin{equation}
f(s,t) = (\pi(f_1(p,t))+\pi(f_2(q,t)),(f_1(p,t),f_2(q,t)))
\end{equation}
where $s=(\pi(p)+\pi(q),(p,q))$. As for a starting state $s_0$, we define it as
\begin{equation}
s_0 = (\pi(s_1)+\pi(s_2),(s_1,s_2))
\end{equation}
Automaticity of $f_1,f_2,\pi$ and of addition in $\mathbb{Q}_2$ ensures that newly defined automatic martingale $f$ is automatic. It is straightforward to check that the fairness condition is satisfied for $f$. Finally, the construction ensures that newly formed setup $d=(f,s_0)$ satisfies the desired relation.\\
\textit{Scalar multiplication}\\
Given a state $p\in S_1$, we define a corresponding state of a desired automatic martingale as $(c\pi(p),p)$. This means that memory of desired automatic martingale given as a state of $f_1$. Having defined state space, we define automatic martingale $f$ and starting state $s_0$ as follows
\begin{equation}
f(s,t)=(c\pi(f_1(p,t)), f_1(p,t))
\end{equation}
where $s=(c\pi(p),p)$. As for a starting state $s_0$, we set
\begin{equation}
s_0 = (c\pi(s_1),s_1)
\end{equation}
Automaticity of $f_1,\pi$ and of scalar multiplication by fixed scalar in $\mathbb{Q}_2$ ensures that the newly formed automatic martingale $f$ is automatic. Again, the fairness condition is clearly preserved. The way construction is done ensures that newly formed setup $d=(f,s_0)$ satisfies desired relations. 
\end{proof}
We can even define infinite sums thanks to the metric topology on real line. Given a countable infinite collection $(c_i)_{i=0}^{\mathbb{N}}$ of positive scalars from $\mathbb{Q}_2$ with bounded sum, i.e. $\sum_i c_i<\infty$ and collection of normed setups $\{d_i\}_{i=0}^{\infty}$ let us define a setup $d = \sum_i c_id_i$ so that:
\begin{equation}
(\sum_{i=0}^{\infty}c_id_i)^{\pi}(Z) = \sum_{i=0}^{\infty}c_id_i^{\pi}(Z)
\end{equation}
for any stream $Z$, where sum and scalar multiplication on the right hand side correspond to point-wise addition of scalar multiplication on $\mathcal{C}$, space of sequences with values in $\mathbb{Q}_2$. Observe that given object still preserves fairness conditions due to elementary properties of limits. On the other hand, given object is no longer automatic, because it involves an infinite amount of data. Nevertheless, this abstract object is going to be useful for us later on. \\
The next theorem asserts the existence of random language given a stringent case of the ordered text. 
\begin{theorem}
There is a random language under $X_{ll}$ text. 
\end{theorem}

\begin{proof}
Let us $(d_i)_{i=0}^{\mathbb{N}}$ be enumeration of all normed setups with possible repetitions. Consider following infinite sum of setups:
\begin{equation}
d = \sum_{i=0}^{\infty}4^{-i}d_i
\end{equation}
We claim that if $d$ does not succeed on an input $X_{ll}\circ L$, then $L$ should be random under $X_{ll}$. To see a reason for this, let us assume that $L$ is not random. Then there is a normed setup $d_j$ succeeding on it, i.e. capital values in $d_j^{\pi}(X_{ll}\circ L)$ are unbounded. For this reason the capital values corresponding to $d$ on the language $L$ should also be unbounded, hence $d$ succeeds on $L$. In the light of this observation, it suffices to constuct a language $L$ on which $d$ happens to fail. Since $d$ satisfies the fairness condition, knowledge of $d(s,t)$ for any datapoint $t$ would allow us to choose membership values for words so that value of capital does not increase from the original $\sum_{i=0}^{\infty}4^{-i} = \frac{4}{3}$. However, $d$ is an infinite object, let alone automatic, so there is no way of directly computing it. One remedy is to use some kind of approximation. Let us define few notions before we proceed with our approximation. Given a regular domain $D$, let us define a function $l:D\to \mathbb{N}$ counting the number of predecessors of the given word $w$ in $D$
\begin{equation}
l(w) = \vert \{v\in D:v\le_{ll}w\}\vert
\end{equation}
With a help of this function, we define an approximation for $d$ given a word $w\in D$
\begin{equation}
d_w = \sum_{i=0}^{l(w)}4^{-i}d_i
\end{equation}
Recall that both $d^{\pi}$ and $d_w^{\pi}$ generate sequences of capital values given some input $Z=X_{ll}\circ L$. Let us compare the values of entries at the position $l(w)$, i.e. the position corresponding to processing of $w$
\begin{align}
d^{\pi}(Z)(l(w)) - d_w^{\pi}(Z)(l(w)) &= \sum_{i=l(w)+1}^{\infty}4^{-i}d_i^{\pi}(Z)(l(w))\\
& \le \sum_{i=l(w)+1}^{\infty}4^{-i}2^{l(w)} \\
&= \frac{1}{3}2^{-l(w)}
\end{align}
The above inequality follows from a simple observation that the capital value increases at most by a factor of two in a single transition. Taking into account this observation, boundedness of $d^{\pi}(Z)$ as a sequence would follow from boundedness of $\{d_w^{\pi}(Z)(l(w)))\}_w$ as a sequence of $w\in D$. In order to achieve latter, we construct the language $L$ following given inductive procedure. Suppose memberships of words up to $w\in D$ in length-lexicographic order, $\{v\in D:v<_{ll}w\}$, has been settled. Next we need to decide a membership of $w$. We wish to ensure that difference between $d_w^{\pi}(Z)(l(w))$ and $d_v^{\pi}(Z)(l(v))$ is rather small, where $v$ is a predecessor of $w$ in $D$ according to length-lexicographic order. We have that
\begin{equation}
d_w^{\pi}(Z)(l(w)) = d_v^{\pi}(Z)(l(w)) + 4^{-l(w)}d_{l(w)}^{\pi}(Z)(l(w))
\end{equation}
Since $d_v$ satisfies the fairness condition, we can choose membership of $w$ in $L$ so that $d_v^{\pi}(Z)(l(w))\le d_v^{\pi}(Z)(l(v))$. Having chosen membership of $w$ in $L$ accordingly, we end up with following:
\begin{align}
d_w^{\pi}(Z)(l(w))&=d_v^{\pi}(Z)(l(w)) + 4^{-l(w)}d_{l(w)}^{\pi}(Z)(l(w))\\  
&\le d_v^{\pi}(Z)(l(v)) + 4^{-l(w)}2^{l(w)}\\
&=d_v^{\pi}(Z)(l(v)) + 2^{-l(w)}
\end{align}
Provided this construction persists in this manner, for any word $w\in D$, we have
\begin{equation}
d_w^{\pi}(Z)(l(w)) \le \sum_{i=0}^{l(w)}2^{-i} \le 2
\end{equation}
Since above expression is uniformly bounded for all $w\in D$, we have achieved our objective. 
\end{proof}
As a remark, let us add that it is possible to analyze computational complexity of the resulting language $L$. Observe that to decide a membership of a word $w$ in $L$, we need to run $l(w)$ many setups on a an input sequence of length $l(w)$. Thus, corresponding time complexity is $O((l(w))^2)$. For any regular domain $D$, $l(w)$ is either of polynomial size or of exponential size with respect to $\vert w \vert$ \cite{Trofimov1981}. This observation leads to following corollary.

\begin{corollary}
Let $L$ be a language obtained as a result of above theorem. Then $L$ is of polynomial time complexity if $D$ is of polynomial growth, and of expoential time complexity if $D$ is of exponential growth. 
\end{corollary}

\subsection{Randomness of collections of languages}
In this subsection we are going to investigate randomness of collections of languages. Similar to the case of individual languages randomness of given collection of languages depends very much on the underlying text class. At first, we are going to consider a randomness of collection of languages under the most general class of texts, $\mathcal{T}$. 
\begin{theorem}
A collection of languages $\mathcal{L}$ is nonrandom under the class $\mathcal{T}$ if and only if $\mathcal{L}\subseteq \mathcal{U}$ for some automatic family $\mathcal{U}$.
\end{theorem}
\begin{proof}
\textit{Forward direction}\\
Let $\mathcal{L}$ is nonrandom under the class $\mathcal{T}$. In other words there is a setup $d=(f,s_0)$ such that $d$ succeeds on every $L\in\mathcal{L}$ under every text $X\in\mathcal{T}$. The idea is similar to the one used in the proof of Theorem 1. For every member $L\in \mathcal{L}$, we try to build an adversarial text, which is supposed to reveal the desired property. We construct the text $X$ inductively as follows. Suppose we have constructed $X$ up to prefix of length $n$, i.e. $X[n]=x_0,\, x_1,\ldots, x_{n-1}$. Let $s_n$ be a state obtained from processing $X\circ L[n]$ from initial state $s_0$. Consider a collection of words which do not increasethe capital at the next stage
\begin{equation}
D_n = \{x\in D\mid \pi(f(s_n,x,L(x)))\le \pi(s_n) \}
\end{equation}
If $D_n$ happens to be nonempty, we append its smallest element, in length-lexicographic order, at the end of $X[n]$. Observe that if $D_n\neq\emptyset$ for infinitely many $n$'s, then we succed to construct a text $X$ such that $\pi(s_n)\le 1$ for all states $s_n$ visited. This contradicts the fact that $d$ succeeds on $L$ under $X$. Hence $D_n = \emptyset$ for some $n$. Let $m$ be the least such such stage, and let $p=s_m$. Then we have
\begin{equation}
\pi(f(p,x,L(x)))>\pi(p) \text{ for all }x\in D
\end{equation}
Thus, we can make use of the fairness condition to check a membership in $L$
\begin{equation}
x\in L \Leftrightarrow \pi(f(p,x,1))>\pi(f(p,x,0))
\end{equation}
Observe that a pair $(L,p)$ identify each other uniquely. We let $p$ to be the index of $L$. Given an index $p$, membership of a word $x\in L_p$ is given by automatic relation above. Hence we are left to identify some regular language containing all necessary indices to satisfy our claim. Observe that $p\in S=C\times M=(\Sigma^*)^k$ for some $k$. Finally we have that, $\mathcal{L}\subseteq \{L_p:p\in S\}$, which concludes this direction. \\
\textit{Converse direction}\\
Let $\mathcal{U}=\{U_e: e\in E\}$ be an automatic family of languages such that $\mathcal{L}\subseteq \mathcal{U}$. It suffices to show nonrandomness of $\mathcal{U}$ under the class $\mathcal{T}$. In other words, we need to show an existence of setup $d=(f,s_0)$ succeeding on any $L\in\mathcal{U}$ under any $X\in\mathcal{T}$. In order to succeed on $L\in\mathcal{U}$, we need to identify its index with respect to $\mathcal{U}$ first. The idea is to employ an idea of enumerative learning, where we go over all indices possible with some repetitions in order to find a correct index. Let $\Gamma$ be an alphabet of the regular language $E$. By introducing a total order on $\Gamma$, we can extend it to the length-lexicographic order on $E$, which also happens to be total and automatic. So we can assume an existence of automatic total order on $E$, called length-lexicographic order. The regular language $E$ plays a role of memory space for our desired martingale, so state space $S=C\times E$. In order to construct a desired setup $d=(f,s_0)$, we set $s_0 =(1,e_0)$, where $e_0= \min_{ll}(E)$. As for $f$, we define it as follows
\begin{equation}
f((c,e),x,b) =  
\begin{cases}
(\frac{3}{2}c,e) & \text{ if } L_e(x)=b\\
(\frac{1}{2}c,succ_E(e)) & \text{ otherwise}
\end{cases}
\end{equation}
where $succ_E(e) = \min_{ll}\{d\in E\mid d>_{ll}e \}$ refers to length-lexicographic successor of $e$ in the language $E$. As a language $L$ belongs to the family $\mathcal{U}$, enumeration process of indices is bound to stabilize, i.e. indices will no longer change from some point onwards. After stabilization occurs, the capital increases by a factor of $\frac{3}{2}$ at each transition. Thus, given setup satisfies desired conditions. 
\end{proof}

\paragraph{The class of repetition-free texts} From the most general class of texts, we move to the smaller class of repetition-free texts, $\mathcal{RF}$. Similar to the case of the class $\mathcal{T}$, we obtain the characterization of randomness for this class. 
\begin{theorem}
A collection of languages $\mathcal{L}$ is not random under the class $\mathcal{RF}$ if only if $\mathcal{L}\subseteq \{U_e\triangle F:e\in E\text{ and } F \text{ is finite}\}$, where $\mathcal{U}=\{U_e:e\in E\}$ is some automatic family.
\end{theorem}

\begin{proof}
\textit{Forward direction}\\	
Suppose that $\mathcal{L}$ is a nonrandom collection of languages under the class $\mathcal{RF}$. To show that $\mathcal{L}$ is of stated form, we employ the strategy similar to the one used in the proof of Theorem 1. Let an automatic setup $d=(f,s_0)$ to realize nonrandomness of $\mathcal{L}$ under $\mathcal{RF}$. Given an arbitrary language $L\in\mathcal{L}$, we are going to construct a text $X\in \mathcal{RF}$ which is going to act against $d$ on $L$. The construction proceeds inductively based on the behaviour of $d$ on the already constructed part. Suppose $X$ has been constructed up to length $n$, $X[n] = x_0,x_1,\ldots,x_{n-1}$. Let $s_n$ be a state resulting in the action of $X\circ L[n]$ based on $d$, where $L\in\mathcal{L}$. Consider a collection of words $D_n$ defined as follows
\begin{equation}
D_n = \{x\in D\mid \pi(f(s_n,x,L(x)))\le \pi(s_n) \text{ and } x\neq x_i \text{ for }i\le n-1 \}
\end{equation}
If $D_n$ happens to be nonempty, we choose the least element in length-lexicographic order to extend existing $X[n]$. If $D_n\neq\emptyset$ for infinitely many $n$'s, then we succeed in constructing the text $X\in \mathcal{RF}$ such that $\pi(s_n)\le 1$ for visited states $s_n$. As it contradicts to initial assumption, $D_n=\emptyset$ for some $n$. Let $m$ be the smallest such stage and $p=s_m$. So, we have
\begin{equation}
x\not\in\{x_0,x_1,\ldots,x_{n-1} \}\Rightarrow \pi(f(s_n,x,L(x)))>\pi(s_n)
\end{equation}
Let us define an approximation to $L$ implied by the above inequality
\begin{equation}
U_p = \{x\in D\mid \pi(f(s_n,x,1))>\pi(f(s_n,x,0)) \}
\end{equation}
a difference between $L$ and $U_p$ is a finite set given by
\begin{equation}
F_p = \{x\in D\mid L(x)\neq U_p(x) \}
\end{equation}
Applying the set difference operator twice, we obtain $L = U_p\triangle F_p$. Since each $p\in S$, we let index set of automatic family to be $E=S=(\Sigma^*)^k$ which is a regular language. Finally, we have
\begin{equation}
\mathcal{L}\subseteq \{U_p\triangle F\mid p\in E,\, F\text{ is finite} \}
\end{equation}
This completes the proof of the forward direction\\
\textit{Converse direction}\\
Let $\mathcal{U}=\{U_e:e\in E\}$ be an automatic family such that $\mathcal{L}\subseteq \{U_e\triangle F: e\in E,\,F\text{ is finite}\}$. To show desired result, it suffices to show nonrandomness of $\mathcal{V}=\{U_e\triangle F\mid e\in E,\,F\text{ is finite} \}$. In other words, we need to construct a setup $d=(f,s_0)$ succeeding on every language $L\in\mathcal{V}$ under any text $X\in \mathcal{RF}$. Again, the idea is to employ enumerative learning approach to find appropriate index of $L$. This time round our task is a bit harder, because even if we might find a right index, a finite difference with a finite set might prevent us from recognizing it. So the trick is to design enumeration procedure so that each index is visited sufficiently often. For this, we set a memory space of desired automatic martingale be a product space $E^2=\{(e,d):e,d\in E\}$. In terms of representation, it corresponds to convoluting elements of $E$. In order to define a desired setup $d=(f,s_0)$, let $s_0=(e_0,e_0)$, where $e_0=\min_{ll}(E)$. As for automatic martingale $f$, it is given as follows 
\begin{equation}
f((c,e,d),x,b) = 
\begin{cases}
(\frac{3}{2}c,e,d) & \text{ if } L_e(x)=b\\
(\frac{1}{2}c,succ_E(e),d) & \text{ if } L_e(x)\neq b \text{ and } e<_{ll}d\\
(\frac{1}{2}c,e_0,succ_E(d)) & \text{ otherwise}
\end{cases}
\end{equation}
Observe that for any element $U_e$ of automatic family and any finite set $F$, there is some word $w$ such that $U_e\triangle F$ and $U_e$ are identical above the given word
\begin{equation}
\forall v\ge_{ll}w \Rightarrow L_e(v)=(L_e\triangle F)(v)
\end{equation}
for this reason, given procedure is bound to stabilize, i.e. a pair of indices will never change from some point onwards. After that the capital value increases by a factor of $\frac{3}{2}$ at each transition. Hence the given setup is successful. This completes the proof of converse direction. 
\end{proof}
\paragraph{Analysis of the collection introduced earlier} In the above theorem, we have encountered the collection of languages of the form $\{L_e\triangle F\mid e\in E,\,F \text{ is finite}\}$ where $\{L_e:e\in E\}$ is an automatic family. One might ask if this notion is any different from the notion of an automatic family. We address this issue in the following theorem.
\begin{theorem}
Let $D$ be a regular domain over $\Sigma$. Let $\{L_e:e\in E \}$ be an automatic family in $D$ and $\mathcal{F}$ be a collection of finite languages in $D$. The following are equivalent
\begin{enumerate}
\item The collection given by $\{L_e\triangle F: e\in E, F\in \mathcal{F} \}$ is an automatic family. 
\item There is a bound $c$ such that for all $n$, $|D\cap \Sigma^{\le n}|\le cn+c$.
\item There is a bound $c$ such that for all $n$, $|D\cap \Sigma^n | \le c$.
\end{enumerate}
\end{theorem}

\begin{proof}
$(1)\to (2):$\\
Let $\mathcal{L} = \{L_e\triangle F\mid e\in E,\, F\text{ is finite} \}$ be an automatic family for some automatic family $\{L_e \}_{e\in E}$. Suppose that $I$ is an index set corresponding to the automatic family $\mathcal{L}$. Observe that for any $U_i\in\mathcal{L}$ and $F$ finite, $U_i\triangle F\in \mathcal{L}$, i.e. $\mathcal{L}$ is closed under symmetric difference with finite sets. Fixing some index $i\in I$, let us consider a following collection of languages
\begin{equation}
\mathcal{H} = \{H_j \}_{j\in I}, \text{ where } H_j = U_i \triangle U_j
\end{equation}
Since an operation of symmetric difference is first-order definable, $\mathcal{H}$ is an automatic family. Moreover, any finite set $F\subseteq D$ belongs to $\mathcal{H}$ for a simple reason that $U_i\triangle(U_i\triangle F) = F$. This means that for any finite set $F$ there is a corresponding index $j\in I$ such that $H_j = F$. Suppose we are given some word $x\in D$. Consider the collection $F_x$ of all finite sets in $D$ consisting of elements no greater than $x$ in length-lexicographic order, i.e. $F_x = \mathcal{P}(\{y\mid y\le_{ll}x \})$. It is clear that for any element of $F_x$ there is a corresponding index $j\in I$. Our aim is to compactify the space of indices corresponding to $F_x$. In other words, we are looking for an index $j(x)$ such that for any $F\in F_x$ there is an index $j$ such that $H_j = F$ with $j\le_{ll}j(x)$. Observe that $j(x)$ can be defined in the first-order fashion as follows
\begin{equation}
j(x)=\min_{ll}\{j\mid \forall k \left[ [\forall y\in H_k \Rightarrow y\le_{ll}x]\Rightarrow [\exists\, r\le_{ll}j\,(H_r = H_k)]\right] \}
\end{equation}
Automaticity of the above mapping gives us
\begin{equation}
\vert j(x) \vert \le \vert x \vert + c
\end{equation}
for some $c$, due to the pumping lemma. Let $x_n$ be lexicographically largest element of $D_{\le n}$. Then $\vert F_{x_n}\vert = 2^{\vert D_{\le n}\vert}$. On the other hand any element of $F_{x_n}$ has a corresponding index no greater than $j(x_n)$. Suppose that $\Gamma$ is an underlying alphabet of $I$. Then ther are at most $\vert \Gamma \vert^{\vert j(x_n)\vert}$ indices  with length no greater than that of $j(x_n)$. Combining two arguments above we get
\begin{equation}
2^{\vert D_{\le n}\vert} = \vert F_{x_n}\vert \le \vert \Gamma \vert^{\vert j(x_n)\vert }\le \vert\Gamma\vert^{n+c}
\end{equation}
Taking $\log$ with base $2$, we obtain that
\begin{equation}
\vert D_{\le n}\vert \le (n+c)\log_2^{\vert \Gamma \vert}
\end{equation}
$(2)\to (3)$:\\
We show contrapositive, i.e. if $\{\vert D_n\vert \}_n$ is unbounded, then $\{\frac{\vert D_{\le n} \vert}{n} \}_n$ is unbounded as well, where $D_n = D\cap \Sigma^n$ and $D_{\le n}=D\cap \Sigma^{\le n}$. Let $\vert D_n \vert = k$ for some $n$. We claim that there is an increasing and unbounded function $f$ such that $\frac{\vert D_{\le m}\vert}{m}\ge f(k)$ for some $m$, uniformly in $k$. Let $p$ be a pumping constant corresponding to the regular domain $D$. Consider prefixes of length $p$ of words belonging to $D_n$. There are at least $\lceil \frac{k}{2^p} \rceil = k_1$ words with the same prefix, say $x$. Let $S$ be a corresponding collection of suffixes extending $x$. By our previous argument, $\vert S \vert \ge k_1$. The regularity of $D$ implies that $x$ can be written as $x=uvw$ such that $x$ and $uv^rw$ are syntactically equivalent for every $r\ge 0$, i.e. for every $y\in \Sigma^*$, $xy\in D\Leftrightarrow uv^rwy\in D$. We set $m=n(p+1)$ and estimate a lower bound for $\frac{\vert D_{\le m}\vert}{m}$. For any $r\le n$ and $y\in S$, $uv^rwy\in D$. Furthermore, each such pair $r,y$ induces a unique word. Thus, $\vert D_{\le m}\vert \ge nk_1$, which implies that
\begin{equation}
\frac{\vert D_{\le m}\vert}{m}\ge \frac{nk_1}{n(p+1)}\ge \frac{k}{2^p(p+1)}=f(k)
\end{equation}
where $f(x)=\frac{x}{2^p(p+1)}$, which is increasing and unbounded. \\
$(3)\to (1):$\\
Due to the first-order definability of symmetric difference operator, it suffices to show that $\mathcal{F}$, the collection of finite subsets of $D$ is an automatic family. Consider some $F\in \mathcal{F}$ and define $F_n = F\cap D_n$. It is clear that $F$ is completely determined by the sequence $\{F_n\}_n$. By the given assumption, there is $c$ such that $\vert D_n \vert \le c$ for all $n$. Given a lexicographic order in $D_n$, each $F_n$ can be associated with a characteristic $c$-tuple, say $(1,0,\ldots,0)$. In a given example of $(1,0,\ldots,0)$, we have that lexicographically least element of $D_n$ belongs to $F_n$, while all other elements do not. By encoding each $c$-tuple with a letter of some alphabet $K$, provided that $\vert K \vert \ge 2^c$, we associate each $c-tuple$ with a letter of $K$ via a map $\phi$. Observe that given encoding is finite, so it can be realized on a finite automaton. Observe that $\phi$ naturally extends to the map from $\mathcal{F}$ to $K^*$ given by
\begin{equation}
\phi(F)=k_0k_1\ldots k_r, \text{ where }\phi(F_n)=k_n
\end{equation}
with $r$ being the largest index such that $F_n\neq \emptyset$. The given extension of $\phi$ can be thought as an indexing function for $\mathcal{F}$. We claim that $\mathcal{F}$ with an index set $K^*$ is an automatic family. Indeed, given a word $x$, with $\vert x \vert =m$ and an index $k = k_0k_1\ldots k_n$, the automata corresponding to the given automatic family checks the value of $k_m$ if exists. Then based on the value of $k_m$, and corresponding $c$-tuple it is possible to determine if $x\in F$ for given $F$. 
\end{proof}
\begin{remark}
Since there are instances of regular domains where $\{\vert D_n\vert \}_n$ is not bounded, take for instance $D=\Sigma^*$, we have that $\{L_e\triangle F\mid e\in E,\, F\text{ is finite} \}$ is not always automatic family for a given automatic family $\{L_e \}_{e\in E}$. 
\end{remark}
The final theorem asserts nonrandomness of the polymial time complexity class $P$ under ordered text, provided that underlying regular domain $D$ is of exponential growth.
\begin{theorem}
Suppose $D$ is a regular domain of exponential growth. Then the class $P$ is nonrandom under the length-lexicographic text, $X_{ll}$. 
\end{theorem}

\begin{proof}
The main idea of the proof is to use exponential gaps provided by length-lexicographic text of $D$ to precompute membership of certain words in a given language $L\in P$. These words need to be arranged so that there is a considerable space between two consecutive elements. More precisely, we construct a sequence of words in $D$, $(a_n)_{n=1}^{\infty}$ such that $\lvert [a_n,a_{n+1}]_{ll} \rvert \ge  \Theta(\gamma^{\vert a_{n+1} \vert})$, with $\gamma>1$, where $\lvert [a_n,a_{n+1}]_{ll}\rvert $ refers to length of a sequence given in length-lexicographic order between elements $a_n$ and $a_{n+1}$. The appearance of the word $a_n$ as an input would signal the machine to start the simulation of $L(a_{n+1})$. One problem is that we do not know the exact algorithm computing membership in $L$, except for the fact that it is of polynomial time complexity. To address this issue, we again use the idea of enumerative learning from algorithmic learning theory. Each time the computed value $U(e,a_{n+1})$ happens to differ from the actual value of $L(a_{n+1})$, we are going to increase an index $e$ in the simulation of an universal Turing machine $U$. From this, it is clear that we have two distinct aims. First, we need to construct elements of desired sequence. Secondly, we need simulate computations of universal Turing machine with subsequent updates of its index. Let us describe our plans of reaching these two aims.
\paragraph{Sequence construction} 
Suppose that at the stage $t$ an element $x$ of sequence arrives as an input, we need to construct a subsequent element $y$ such that $\lvert [x,y]_{ll}\rvert\ge \Theta(\gamma^{\vert y\vert})$. The idea is to choose the smallest string of length at least $t$ in the domain $D$. So:
\begin{equation}
y = \min_{ll}\{z\in D\mid\vert z\vert\ge t\}
\end{equation}
Due to pumping lemma, we have that $\vert y \vert \le t + p$, where $p$ is a pumping constant corresponding to a regular domain $D$. Since $D$ is of exponential growth, there is an integer $k$ such that $\lvert D_{<nk} \rvert \ge 2^n$ for all $n$. Given this, we compute $l(y)$, number of predecessors of $y$ in length-lexicographic order in $D$. We have that:
\begin{equation}
2 ^ {\frac{t - k}{k}}\le l(y)  \le 2^{t+p}
\end{equation} 
Using this inequality, we obtain that:
\begin{equation}
\lvert [x,y]_{ll}\rvert\ge 2^{\frac{t-k}{k}} - t = \Theta(\gamma^{\vert y \vert})
\end{equation}
with $\gamma = 2^{\frac{1}{k}}>1$. Thus, this procedure succeeds in generating a sequence of words satisfying given condition. 
\paragraph{Simulation} Having computed next element $y$ of the sequence, we wish to simulate a computation of $L(y)$ for a given language $L$. We employ enumerative learning, whereby we go over indices of Turing machines, until we find an appropriate one. Given current hypothesis $e$ as an index, we compute $U(e,y)$, where $U$ is an universal Turing machine. Since step-wise transition of universal Turing machine can be realized as an automatic function, we simulate the computation of $U(e,y)$ until word $y$ arrives as an input. In case the computation terminates with an output of $0$ or $1$, we check this output against the corresponding label $L(y)$. If they happen to be equal, we keep the current index. If not, we update current index to its length-lexicographic successor, $d$. Now we present automatic representations of these two procedures. 
\paragraph{Automatic presentation} 
Having described our plan, we need to supply its automatic presentations. A memory space of automatic setup will given as a product of two spaces, $M=M_1\times M_2$, where $M_1$ will be used for sequence construction and $M_2$ will be used for simulation. \\
$M_1$ consists of two values: a counter value $m_1^C$ and a string value $m_1^S$. At each transition, the counter value is appended with $0$ at its end. This ensures that the length of $m_1^C$ at stage $t$ equals to $t$. The string value, $m_1^S$, holds the value of string, $y$, under simulation. When a string $y$ appears as an input $(y,L(y))$, it causes activation of $M_1$ which updates its string value as follows
\begin{equation}
m_1^S = \min_{ll}\{z\in D\mid \vert z \vert \ge m_1^C \}
\end{equation}
After this $M_1$ operates as per normal. Let $m_1 = (\varepsilon, d_0)$ serve as a first component of a starting state coming from $M_1$, where $d_0 = \min_{ll}(D)$.\\
As for $M_2$, it is intended to simulate computations of the universal Turing machine mwntioned earlier. To achieve this end, we represent $M_2$ as a convolution of five standard parts of universal Turing machine, namely:
\begin{itemize}
\item Input, $m_2^I$;
\item Hypothesis, $m_2^H$;
\item State of universal machine, $m_2^S$;
\item Work tape, $m_2^W$;
\item Output tape, $m_2^O$.
\end{itemize}
An input $m_2^I$ holds a value of a string $y$ under simulation. A hypothesis value holds the value of current hypothesis for index, $e$. The state of universal Turing machine $m_2^S$ and the work tape $m_2^W$ are to be interpreted in a standard manner. Finally, an output tape $m_2^O$ holds the result of underlying simulation, if ever computed. At stages when $M_1$ is not activated, $M_2$ proceeds with the simulation assigned to it. Observe that since transitions of $U$ can be realized as an automatic function, automatic martingales are capable of such simulations. At the stage when $M_1$ is activated, $M_2$ compares produced output, if any, with $L(y)$. If they happen to be same, then automatic martingale $f$ assigns $\frac{3}{2}$ of its capital to the outcome corresponding to $m_2^O$. If they happen to be different, then automatic martingale $f$ updates index $e$ to its length-lexicographic successor $d$, making a trivial bet with ratio $1:1$, as it does for non-activation phases. Let $m_2 = (d_0, \varepsilon, \varepsilon, \varepsilon, \varepsilon)$ serve as a second component of a starting state coming from $M_2$. Then the starting state for desired automatic setup is $s_0=(1,m_1,m_2)$. 
\paragraph{Verification} Let us provide a short verification argument for above construction. Let the time complexity of deciding membership in a given language $L$ be $O(p(n))$ for some polynomial $p(n)$. Then the universal Turing machine $U$ simulates the same algorithm with at most $O(p^2(n))$ time complexity. In above construction, allowed time frame for a computation with input of length $n$ is at least $\Theta(\gamma^n)$ for some $\gamma>1$. As an exponential function grows faster than any polynomial function, there is some integer $N$ such that for any input of length $n\ge N$, the time frame provided will suffice for a computation provided a correct index. One caution might be the fact that by this time we might pass some of correct indices for $L$ in an enumeration process. However, one should recall that there are infinitely many indices corresponding to $L$. This shows that given automatic martingale will not make any incorrect bets from some time onwards. Given verification process concludes our construction of desired setup $d=(f,s_0)$. Let us remark that given argument is similar to a finite injury argument from computability theory. 
\end{proof}

\section{Discussion}
\subsection{Summary}
In this paper we investigated randomness of formal languages. In order to do so, we have defined a notion of 
automatic maritngales betting on a sequence of words. With this notion in our hands, we have defined randomness of languages both individually and collectively. We have observed a strong dependence of randomness on underlying class of texts. For example,
a language $L$ is random under the class $\mathcal{T}$ if and only if it is not regular. On the other hand, any computable language is nonrandom under some dynamic text. These results show a wide randomness spectrum that formal languages occupy depending on the text class under consideration. 
\subsection{Open Questions}
This paper tried to pursue a new approach in determinining randomness of formal languages. We have 
analyzed few properties of automatic martingales and random languages, but there seems to be much more
questions lying ahead of us. In this sense, we might have scratched the surface of large body of knowledge on 
random formal languages. To give examples of questions we want to know an answer for, let us provide few
immediate open questions
\begin{itemize}
\item Randomness of context-sensitive languages: given a context-sensitive 
language $L$ in some regular domain $D$, what can we say about its randomness 
under the length-lexicographic text, $X_{ll}$?
\item Are the results given in the paper invariant under arbitrary representations of the capital? 
For our purposes, we used dyadic rationals as capital values. What happens if 
one uses different representation for capital values?
\end{itemize}

\section*{Acknowledgement}
We would like to thank Frank Stephan for many discussions on topics of algorithmic randomness, algorithmic learning theory and automata theory.

\end{document}